\documentclass[11pt]{article}

\makeatletter
\renewcommand\section{\@startsection {section}{1}{\z@}%
                                   {-3.5ex \@plus -1ex \@minus -.2ex}%
                                   {2.3ex \@plus.2ex}%
                                   {\normalfont\bfseries}}

\def\@maketitle{%
  \newpage
  \null
  \vskip 2em%
  \begin{center}%
  \let \footnote \thanks
    {\large \@title \par}%
    \vskip 1.5em%
    {\normalsize
      \lineskip .5em%
      \begin{tabular}[t]{c}%
        \@author
      \end{tabular}\par}%
    \vskip 1em%
  \end{center}%
  \par
  \vskip 1.5em}
\makeatother
\usepackage{tikz}
\usetikzlibrary{automata} 
\usetikzlibrary{petri}
\usepgflibrary{shapes.geometric}

\newtheorem{The}{Theorem}
\newtheorem{Lem}[The]{Lemma}
\newtheorem{Cor}[The]{Corollary}
\newtheorem{Conj}[The]{Conjecture}
\newtheorem{Exa}{Example}

\newtheorem{Def}[The]{Definition}
\newenvironment{proof}{\noindent\textbf{Proof: }}{\hfill \small $\Box$}

\newtheorem{Pro}[The]{Proposition}
\usepackage{latexsym}
\usepackage{amssymb,xspace} 
\usepackage{epsfig}

\usepackage{graphics}

\DeclareGraphicsExtensions{{}}
\newcommand{\lra}{\longrightarrow}
\newcommand{\raw}{\rightarrow}

\newcommand{\Rawe}{\Rightarrow\!\!\!^*\;}

\newcommand{\prt}[1]{\langle #1\rangle}

\newcommand{\sumat}[1]{\sum\limits_{#1}}

\newcommand{\maL}{\mathcal{L}}

\newcommand{\vdashe}{\;\;^*\!\!\!\!\!\!\vdash\;}

\newcommand{\nla}{{\sc nla}\xspace}
\newcommand{\dla}{{\sc dla}\xspace}
\newcommand{\pda}{{\sc pda}\xspace}

\newcommand{\nfa}{{\sc nfa}\xspace}
\newcommand{\slnf}{{\sc slnf}\xspace}
\newcommand{\lnf}{{\sc lnf}\xspace}

\begin{document}

\title{{\bf {\Large Some Subclasses of Linear Languages based on Nondeterministic Linear Automata}} \\
\vspace{.5cm} {\bf  Benjam\'in Bedregal} \\ 
\small Departamento de Inform\'atica e Matem\'atica Aplicada, \\ 
\small Universidade Federal do Rio Grande do Norte \\ 
\small \textit{bedregal@dimap.ufrn.br} 
}

\maketitle

\begin{abstract} 
In this paper we consider the class of
$\lambda$-nondeterministic linear automata as a model of the class of linear
languages. As usual in other automata
models, $\lambda$-moves do not increase the acceptance power. The main contribution of this paper
is to introduce the deterministic linear automata and
even linear automata, i.e. the natural restriction of
nondeterministic linear automata for the deterministic and even
linear language classes, respectively. In particular, there are different, but not equivalents,
proposals for the class of ``deterministic'' linear languages. We proved here that the class of 
languages accepted by the deterministic linear automata are not contained in any of the these classes and 
in fact they properly contain these classes. Another, contribution is the generation of an infinite hierarchy of formal languages, 
going from the class of languages accepted by deterministic linear automata and achieved, in the limit, the class of linear languages.
\end{abstract}
\vspace{12pt}\noindent{\bf Keywords:} 
  Linear languages, nondeterministic linear automata, deterministic linear automata, deterministic linear 
  languages, even linear languages, $\lambda$-moves and degree of explicit nondeterminism.


\section{Introduction}

The class of linear languages, also known by linear context-free
languages, is situated in the  extended Chomsky hierarchy  between
the classes of regular and context-free languages. Linear grammars 
(the more well known model for the class of linear languages) allow to have a control on 
matches between leftmost symbols with rightmost symbols of
a substring. This capability allow that the most of the more commonly used
examples of context-free languages, which are not regular languages, belong to this class, as for example 
the languages of palindromes and $\{a^nb^n: n\geq 1\}$. An example of a context-free
language which is not linear is $\{ a^mb^ma^nb^n\; :\; n,m\geq
1\}$ \cite{ABB97}.



The most usual models for linear languages are the linear  grammars and
some normal form for them have been proposed. In terms of automata counterpart, it is
possible to find at least six different models: Two-tape
nondeterministic finite automata   of a special type
\cite{Har78,Ros67},  finite transducer \cite{Mak03,Ros67},
one-turn pushdown automata \cite{ABB97,GS66,Har78,HE04,HU79},
right-left-monotone restarting automata \cite{Jur06}, $\lambda$-nondeterministic linear
automata \cite{Bed10,Bed11} and zipper finite-state automata \cite{SW11}.
Despite the merits of the four first models, they are not intrinsic, in the sense that
their definition includes external elements. 
For example, in the
two-tape nondeterministic finite automata of a special type, the
input is split into two tapes with the second one containing the
reverse of the right part of the input string.
Notice that,  all
the usual types of automata -- e.g. finite automata, pushdown
automata (\pda in short), Turing machines, etc. -- assume that the
input string is sequentially disposed, without any modification, in the input tape at the start of
execution that here we call natural condition. Thus, this model 
does not satisfy the natural condition 
and it is not intrinsic because the choice of how and where to divide
the input and the application of the reversion is external to the
model.

Analogously, the finite transducer model also presupposes
that externally the input string is divided into two parts which are
separated by a special symbol and where the second part is
reverted. Therefore, this model also presupposes a previous
knowledge of where the input must be divided. Thus, this model does not satisfy either the 
natural condition and it considers external agents. On the other hand, a
turn in a \pda is a move which decreases the stack and which is
preceded by a move that increases the stack, the one-turn \pda
model (a \pda which makes at most one turn in the stack).
Determining the maximal possible quantity of turns that a \pda can
make (i.e. consider all possible input strings) requires an
external control.  
Thus, one-turn \pda model is not intrinsic. Finally, the
right-left-monotone restarting automata introduced in \cite{Jur06}
is a restriction of a more general class of automata, the
restarting automata introduced in \cite{Jan95,Jan99}.    Since in
this automata the transitions work with strings instead of symbols
and they allow actions such as the reversion of strings, this model of
linear language is higher leveled and also more
complex than the other models. On the other hand, this model
requires that at each cycle the distance from the actual place,
where a rewrite takes place, to the right and left end of the tape
must not increase. However the verification of that condition can
not be made internally in the model, and therefore it is also a non intrinsic model.

The  $\lambda$-nondeterministic linear
automata, $\lambda$-\nla, which extends  the notion of nondeterministic finite
automata with $\lambda$-moves, $\lambda$-\nfa in short, and
therefore properly contains this class of automata. This
automata model, in our viewpoint, does not consider external elements,
it is not a subclass of automata used to model a broader class of
languages, it is  strongly inspired in a formal
norm introduced here for linear grammars and therefore it is simpler than the
previous models. 

The last model, zipper finite-state automata \cite{SW11}, is similar to the the previous one, in the 
sense that it read the input string from the left and from the right side, but differently of the $\lambda$-\nla, 
it is made simultaneoulsy and every moviment can read more than one symbols of each  side of the string.


In addition, the $\lambda$-moves in the usual classes of
automata are not essential for the models. Thus, for example,
$\lambda$-moves do not increase the acceptance power
of finite automata (see, for example, \cite{HU79}), of pushdown
automata (see, for example, \cite{Har78}) and of Turing machines
(by Church-thesis). Analogously, in \cite{Bed11} it was proved that 
$\lambda$-moves do not  increase either the acceptance power of $\lambda$-nondeterministic linear automata.

In this work we introduce the
deterministic version of this automata, called deterministic linear automata, \dla in short,  and compare
the class of languages accepted by \dla's with the class \textbf{DL} introduced by de la
Higuera and Oncina in  \cite{HO02} and with the class of linear
languages and deterministic context-free languages. We also
provide a characterization of the \nla which are equivalent to some \dla.
In  this paper we also determine an infinite hierarchy of  classes of formal languages which are among
 the class of languages accepted by \dla's and the class of linear languages. Finally, we show how the 
 class of even linear languages, introduced originally by Amar and Putzolu in
\cite{AP64},  can be captured in the nondeterministic linear automata model.

\section{Linear grammars}\label{sec-LG}

As usual a formal grammar $G$  is a tuple $\prt{V,T,S,P}$ where
$V$ is a finite set of variables, $T$ is a set of terminal symbols
and therefore $V\cap T=\emptyset$, $S\in V$ is the start variable
and $P\subseteq (V\cup T)^+\times (V\cup T)^*$ is the set of
productions. Ordered pairs $(x,y)\in P$ are denoted by
$x\rightarrow y$.

A grammar $G=\prt{V,T,S,P}$ is \textbf{linear}, if  each
production in $P$ has a variable in its left side  and has at
most  one variable in its right side, without restriction in the
position of this variable. A variable $A\in V$ will be called
\textbf{left linear} if in each production $A\rightarrow y\in P$,
either $y\in T^*$ or $y=Bz$ for some $z\in T^*$ and $B\in V$.
Analogously, a variable $A$ will be called \textbf{right linear}
if in each production $A\rightarrow y\in P$  either $y\in T^*$ or
$y=zB$ for some $z\in T^*$ and $B\in V$. A linear grammar $G$ is
in \textbf{linear normal form}, in short  \lnf, if each variable
$A\in V$ is left or right linear. Thus a variable $A$ in a linear
grammar is both left and right linear, just when  each production
having $A$ in the left side has at the right side either a string
of terminal symbols (including the empty string) or a single
variable.

Notice  that our linear normal form is subtly different from the
usual linear normal form (e.g. see \cite{Har78,Lin01,Sal73}) where
it is tolerated two productions with the same variable in the left
side, but with its right side containing a variable in the
leftmost and rightmost positions, respectively.

\begin{Exa} \label{ex-LG}
The grammar  $G=\prt{\{S\}, \{a,b\},S,P}$ where $P$ is given by

$$S\raw aSb \mid aSbb \mid aSbbb \mid ab \mid abb  \mid abbb $$

\noindent is linear but it is not in the \lnf.

On the other hand, the grammar  $G=\prt{\{S\}, \{a,b\},S,P}$ where $P$ is given by

$$S\raw aA \mid aA \mid aA \mid ab \mid abb  \mid abbb $$

$$A\raw Sb \mid Sbb \mid Sbbb$$

is in the \lnf.
\end{Exa}




As usual, for each $u,v,w\in (V\cup T)^*$ and $A\in V$,
$uAw\Rightarrow uvw$ if there is a production $A\raw v\in P$.
Let $\Rawe$ be the reflexive and transitive closure of
$\Rightarrow$. The \textbf{language generated} by a linear grammar
$G$ is

$$\maL(G)=\{w\in T^*\; :\; S\Rawe w\}$$

In Example \ref{ex-LG},

\begin{equation} \label{eq-NLL-ex}
 \maL(G)=\{a^mb^n\; :\; 1\leq m\leq n\leq 3m\}
\end{equation}



Languages generated by linear grammars are called \textbf{linear
languages}.

\begin{Lem}\label{lem-LNF} \cite{Bed11}
Let $G$ be a linear grammar. Then there exists a linear grammar
$G'$ in the \lnf such that $\maL(G)=\maL(G')$.
\end{Lem}

%
%

\begin{Exa} \label{ex-LNF}
The \lnf of $G$ in Example \ref{ex-LG} obtained following the
algorithm given in the proof of Lemma \ref{lem-LNF} given in \cite[Lemma 2.1]{Bed11}  is the grammar
$G'=\prt{\{S,A\}, \{a,b\},S,P'}$, where $P'$ is given by

$\begin{array}{l}
S\raw aA \mid ab \mid abb \mid abbb \\
A\raw Sb \mid Sbb \mid Sbbb
\end{array}$



\end{Exa}

A linear grammar $G$ is  in  \textbf{strong linear normal form},
in short  \slnf, if it is in \lnf and the right side of each
production is of the form $aA$ or $Aa$ where $a\in
T\cup\{\lambda\}$ and $A\in V\cup\{\lambda\}$.

\begin{Pro} \label{pro-SLNF} \cite{Bed11}
Let $G$ be a linear grammar. Then there exists a linear grammar
$\widehat{G}$ in \slnf such that $\maL(G)=\maL(\widehat{G})$.
\end{Pro}

%
%

\begin{Exa} \label{ex-SLNF}
The \slnf of $G'$ in Example \ref{ex-LNF} obtained following the
algorithm given in the proof of Proposition \ref{pro-SLNF}  given in \cite[Prop. 2.1.]{Bed11} is the grammar
$\widehat{G}=\prt{\{S,A,B,C,D,E,F\}, \{a,b\},S,\widehat{P}}$ where
$\widehat{P}$ is given by

$\begin{array}{l}
S\raw aA \mid aD \\
A\raw Sb \mid Bb \\
B\raw Sb \mid Cb \\
C\raw Sb \\
D \raw b \mid bE \\
E\raw b \mid bF \\
F\raw b \\
\end{array}$


\end{Exa}

\subsection{Deterministic Linear Grammars}

A grammar $G=\prt{V,T,S,P}$ is a {\bf deterministic linear grammar} if each
production in $P$ has the form $A\rightarrow aBu$ or $A\rightarrow
\lambda$ and for every $a\in T$, $A,B,C\in V$ and $u,v\in T^*$, if $A\rightarrow aBu$, $A\rightarrow aCv\in P$ then
$B=C$ and $u=v$. 

\begin{Lem}\label{lem-DLNF}
Let $G$ be a deterministic linear grammar. Then there exists a deterministic linear grammar
$G'$ in the \lnf such that $\maL(G)=\maL(G')$.
\end{Lem}

\begin{proof} For each $A\in V$ and $a\in T$ let $A_a=\{A\raw aBv\in P
\; :$ for some $B\in V$ and $v\in T^+\}$. If $A_a\neq\emptyset$
then substitute in $P$ each production $A\raw aBv$  by the
productions $A\raw aC$ and $C\raw Bv$, where $C$ is a new variable
and  add $C$ to $V$. At this moment each production has the form

$$A\raw aB, A\raw Bu\mbox{ or }A\raw u$$

\noindent for some $a\in T$, $u\in T^*$ and $A,B\in V$. If $A\in V$ is
neither right linear nor left linear variable, then substitute
each production $A\raw Bv \in P$  by the productions  $A\raw C$
and $C\raw Bv$, where $C$ is a new variable and add $C$ to $V$.
The resulting grammar clearly is deterministic, is in \lnf and is equivalent with
the original.
\end{proof}

 \begin{Pro} \label{pro-DSLNF} 
Let $G$ be a deterministic linear grammar. Then there exists a  deterministic linear grammar
$\widehat{G}$ in \slnf such that $\maL(G)=\maL(\widehat{G})$.
\end{Pro}

 \begin{proof}
Let $G$ be a deterministic linear grammar. By the Lemma \ref{lem-DLNF}, we can suppose without loss of generality that 
 $G$ is in the \lnf. Now, apply the next algorithm:
While exists $A\in V$ and $a\in T$ such that $A_a=\{A\raw ay\in P\; :\; {\mid} y{\mid}\geq 2
\mbox{ or }y\in T\}\neq \emptyset$,
 change each $A\raw ay\in A_a$ by the productions $A\raw aB$ and
 $B\raw y$, where $B$ is new and add $B$ to $V$. Lately, in
 analogous way, while there are $A\in V$ and $a\in T$ such that
 $A^a=\{A\raw ya\in P\; :\; {\mid} y{\mid}\geq 2\mbox{ or }y\in
 T\}\neq \emptyset$, change each $A\raw ya\in A^a$ by the
 productions $A\raw Ba$ and $B\raw y$, where $B$ is new  and add
 $B$ to $V$.
 The result is a deterministic linear grammar in \slnf equivalent with $G$.
\end{proof}

\begin{Exa}
 Let $G=\prt{\{S,A\}, \{a,b\},S,P}$ the determinitic linear grammar  where
$P$ is given by

$\begin{array}{l}
S\raw bbS \mid aAb \\
A\raw aaAbb \mid \lambda \\
\end{array}$
In this case $\maL(G)=\{b^mab: m$ is even and $n$ is odd$\}$. Their \lnf following the alghorithm in the Lemma 
\ref{lem-DLNF} is the deterministic linear grammar $G'=\prt{\{S,A,B,C,D,E\}, \{a,b\},S,P'}$  where
$P'$ is given by

$\begin{array}{l}
S\raw bB \mid aC \\
A\raw aD \mid \lambda \\
B\raw bS \\
C\raw Ab  \\
D\raw aE \\
E\raw Abb \\
\end{array}$

Now applying the algorithm in Proposition \ref{pro-DSLNF} we obtain the  deterministic linear grammar 
$\widehat{G}=\prt{\{S,A,B,C,D,E,F\}, \{a,b\},S,\widehat{P}}$ where $\widehat{P}$ is $P'$ by substitute the 
production $E\raw Abb$ for $E\raw Fb$ and $F\raw Ab$. Thus, clearly, $\widehat{G}$ is in \slnf.
\end{Exa}

\section{$\lambda$-Nondeterministic linear
automata}\label{sec-NLA}

A $\lambda$-{\bf nondeterministic linear automata}, $\lambda$-\nla
in short, consist of two disjoint finite sets of states ($Q_L$ and
$Q_R$) some of which will be considered as accepting states, an
input tape which is divided into cells and it is not limited at the
right, each cell can hold a symbol from a finite input alphabet,
two read heads and a control unit which  manages the behavior of
the $\lambda$-\nla in accordance with the current configuration. The
execution of a $\lambda$-\nla starts with
 a string in the input tape, with the left read head pointing to
 leftmost symbol, the right read head pointing to the rightmost symbol and
 the current state being a state of a special set of states of $Q_L\cup Q_R$ called set of start 
 states\footnote{The use of a set of start states, although not being usual, 
 had been used in several models of automata. For example, \cite[Def.4.1]{CL89} and 
 \cite[page 32]{KMT97} in non-deterministic finite automata and 
 \cite[page 89]{Coh91}, \cite[page 52]{Har78} in transition systems.}.
A computation step in a $\lambda$-\nla is made as follows: the
control unit, depending on the class which is belonged the current state, uses the left or the right read head to scan
a symbol from the tape, moves the left read head one cell to the right
if the current state is in $Q_L$ or moves the right read head one cell
to the left if the current state is in $Q_R$, and by making a
nondeterministic choice, it changes the state choosing it from a set
of possible states. The control unit of a $\lambda$-\nla also 
allows changing of state without moving the read heads, i.e. without reading
an input symbol. The computation halts when a read head passes over the
other read head or when there is no choice of possible actions. A string 
is only accepted when  one read head passes over the other read head and the current state is an accepting state.

Figure \ref{fig-gen-NLA} illustrates a schematic
representation of a $\lambda$-\nla.


\begin{figure}[h]
\begin{center}
\includegraphics[height=5cm,width=10cm]{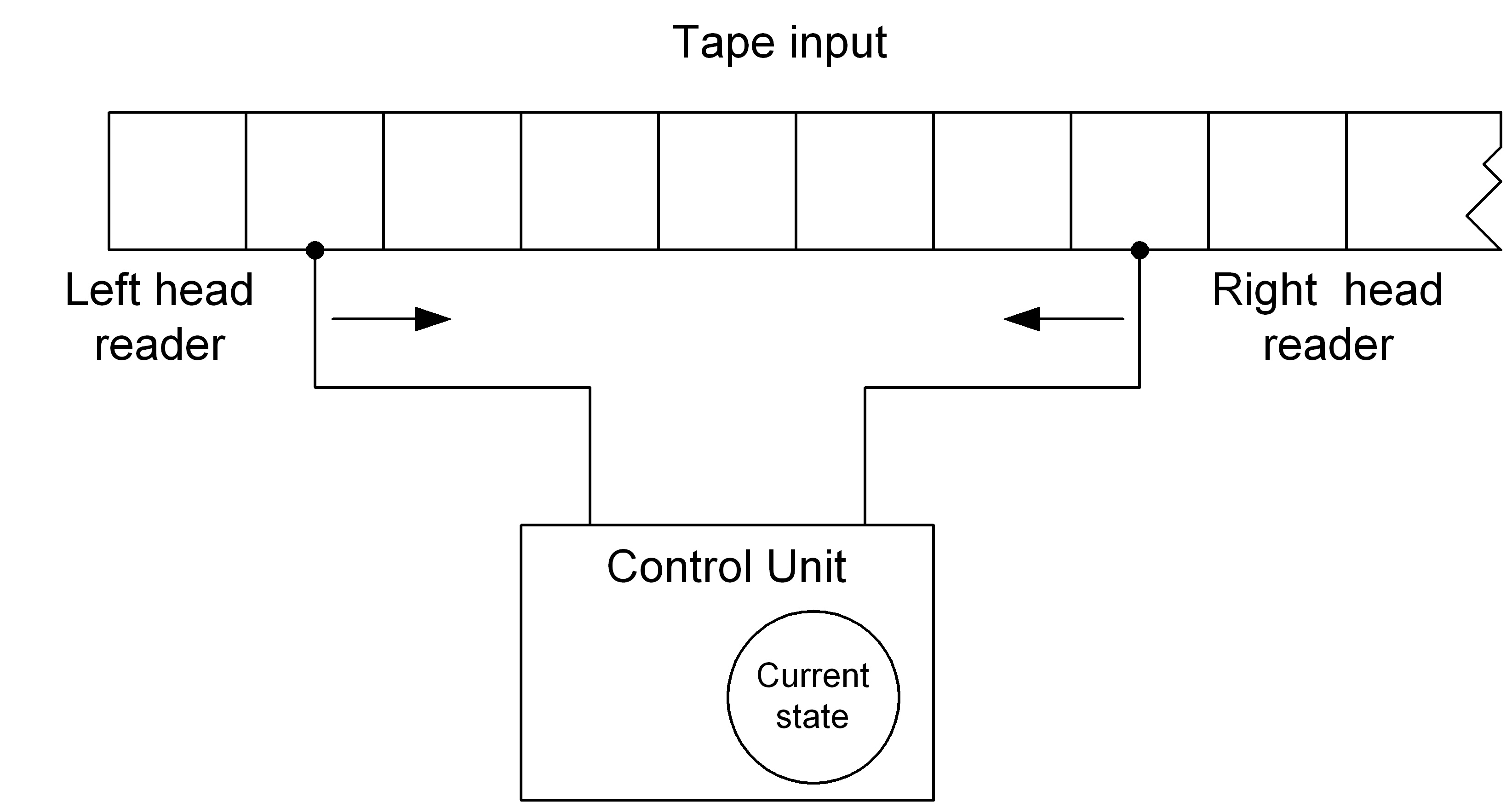}
\caption{Schematic representation of a $\lambda$-nondeterministic
linear automata.}
 \label{fig-gen-NLA}
 \end{center}
\end{figure}

 Formally, a $\lambda$-\nla is a
sextuple $M=\prt{Q_L,Q_R,\Sigma,\delta,I,F}$ where $Q_L$ and
$Q_R$ are disjoint and finite sets of states, $\Sigma$ is a finite
set of input symbols (the alphabet), $I\subseteq Q_L\cup Q_R$ is the set of  start
states, $F\subseteq Q_L\cup Q_R$ is the set of final or accepting
states and $\delta: (Q_L\cup Q_R)\times
(\Sigma\cup\{\lambda\})\lra \mathcal{P}(Q_L\cup Q_R)$.

Analogously to finite automata, each $\lambda$-\nla has associated to itself
a directed graph, called transition diagram. In order to
distinguish the states in $Q_L$ from the states in $Q_R$ we use circles and squares to represent them.

Notice that, if $M$ is a $\lambda$-\nla with $Q_R=\emptyset$ (or $Q_L=\empty$), then
$\maL(M)$ is a regular language. In fact, in this case,
$M'=\prt{Q_L,\Sigma,I,\delta,F}$ is a $\lambda$-nondeterministic
finite automaton (with a set of start states) such that $\maL(M')=\maL(M)$ and whose
transition diagram is exactly the same as the transition diagram
for $M$. Thus, $\lambda$-\nla is a natural extension of
$\lambda$-nondeterministic finite automata.

\begin{Exa} \label{ex-NLA}
Figure \ref{fig-NLA-ex2} illustrates the $\lambda$-\nla
$M=\prt{Q_L,Q_R,\Sigma,\delta,\{q_0\},F}$ where

\begin{itemize}
\item $Q_L=\{q_0,q_1,q_2,q_3\}$

\item $Q_R=\{p_1,p_2,p_3,p_4\}$

\item $\Sigma=\{a,b\}$

\item $F=\{p_1,q_2\}$

\item $\delta(q_0,a)=\{q_0,p_1\}$, $\delta(q_0,\lambda)=\{p_3\}$, $\delta(p_1,a)=\{p_2\}$,
$\delta(p_2,a)=\{q_1\}$, $\delta(q_1,b)=\{p_1\}$,
$\delta(p_3,b)=\{p_3,q_2\}$, $\delta(q_2,b)=\{q_3\}$,
$\delta(q_3),b)=\{p_4\}$, $\delta(p_4,a)=\{q_2\}$ and empty for the
remain (i.e. $\delta(q_0,b)=\emptyset$, $\delta(p_1,b)=\emptyset$,
etc.).
\end{itemize}
\end{Exa}



%

\begin{figure}
\centering
\begin{tikzpicture}[shorten >=1pt,node distance=2cm,auto]

\node[state,initial]  (q_0)                          {$q_0$};
\node[transition, regular polygon sides=\4, minimum size=0.9cm]          (p_1) [right of=q_0]     { $\;\;P_1\;\;$};
\node[transition, regular polygon sides=\4, minimum size=0.8cm]          (p1) [right of=q_0]     { };
\node[transition, regular polygon sides=\4, minimum size=0.9cm]          (p_2) [right of=p_1]     { $\;\;P_2\;\;$};
\node[state]  (q_1)       [right of=p_2]                   {$q_1$};
\node[transition, regular polygon sides=\4, minimum size=0.9cm]          (p_3) [below of=q_0]     {$\;\;P_3\;\;$ };
\node[state,accepting]  (q_2)       [right of=p_3]                   {$q_2$};
\node[state]  (q_3)       [right of=q_2]                   {$q_3$};
\node[transition, regular polygon sides=\4, minimum size=0.9cm]          (p_4) [right of=q_3]     { $\;\;P_4\;\;$};


\path[->] (q_0) edge                 node               {\small a}     (p_1)
          (q_0) edge [loop above]    node               {\small a}     (q_0)
          (q_0) edge                 node               {\small $\lambda$}     (p_3)
          (p_1) edge                 node               {\small a}     (p_2)
          (p_2) edge                 node               {\small a}     (q_1)
          (q_1) edge  [bend right]   node [above]       {\small b}     (p_1)
          (p_3) edge  [loop below]    node [below]       {\small b}     (p_3)
          (p_3) edge                node               {\small b}     (q_2)
          (q_2) edge                node               {\small b}     (q_3)
          (q_3) edge                 node               {\small b}     (p_4)
          (p_4) edge    [bend left] node  [below]     {\small a}     (q_2);
\end{tikzpicture}
\caption{Transition diagram of a $\lambda$-nondeterministic linear automata.}\label{fig-NLA-ex2}
\end{figure}
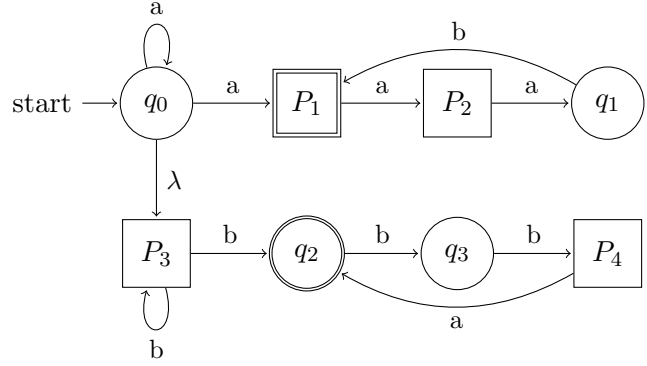

An \textbf{instantaneous description} (ID) of a $\lambda$-\nla
must record the current state, the string remaining to read and which read 
head is active. Thus an ID is a pair $(q,w)$ in $(Q_L\cup
Q_R)\times \Sigma^*$  meaning that it remains to read the string $w$,
the current state is $q$ and the read head which is active is the left
when $q\in Q_L$ and is the right when $q\in Q_R$.

The symbol $\vdash_M$ denotes a \textbf{valid move}  from an ID to
another ID in a $\lambda$-\nla $M$. When the subscript $M$ is clear we omit it.
Thus, for each $q\in Q_L$, $q\prime\in Q_L\cup Q_R$,
$p\in Q_R$, $w\in \Sigma^*$ and $a\in\Sigma\cup\{\lambda\}$

$(q,aw)\vdash (q\prime,w)$ is possible if and only if
$q\prime\in\delta(q,a)$ and

$(p,wa)\vdash (q\prime,w)$ is possible if and only
$q\prime\in\delta(p,a)$

We use $\vdashe$ for the reflexive and transitive closure of
$\vdash$, i.e. $\vdashe$ represents moves involving an arbitrary
number of steps. Thus, in Example \ref{ex-NLA}, we
have that $(q_0,abbaaaa)\vdashe (p_1,baa)$, because

$$(q_0,abbaaaa)\vdash (p_1,bbaaaa,)\vdash (p_2,bbaaa)\vdash
(q_1,bbaa)\vdash (p_1,baa)$$

The language accepted by a $\lambda$-\nla $M$ is the set

$$\maL(M)=\{w\in\Sigma^*\; :\; (q_0,w)\vdashe (q_f,\lambda)\mbox{ for some }q_0\in I\mbox{ and }q_f\in
F\}$$

For the case of $\lambda$-\nla $M$ in Example \ref{ex-NLA},

$$\begin{array}{ll}
\maL(M) = & \{a^mb^na^{2n}\; :\; m\geq 1\mbox{ and }n\geq 0\}\cup \\
& \{ a^kb^{2m}a^mb^n \; :\; k,m\geq 0\mbox{ and } n\geq 1\}. \\
\end{array}$$

Notice that, the halting mechanism of the \nla,  despite not being 
 explicit in its mathematical formulation, can be formalized
by the ID notion as follows: when an ID $(q,\lambda)$ is  achieved  
through a move (e.g. $(q',a)\vdash (q,\lambda)$), we are in the
situation of ``a read head passes over the other read head''.

\subsection{$\lambda$-\nla and linear languages}

First, we will prove that each language accepted by some
$\lambda$-\nla $M$ is linear, i.e.

\begin{The} \label{teo-NLA-LG} \cite{Bed11}
Let $M=\prt{Q_L,Q_R,\Sigma,\delta,I,F}$ be a $\lambda$-\nla.
Then there exists a linear grammar $G$ such that $\maL(M)=\maL(G)$.
\end{The}

%
%

%
%
%
%
%
%
%
%
%
%

Conversely, each language
generated by a linear grammar $G$ is accepted by some
$\lambda$-\nla.

\begin{The} \label{theo-SLNF-NLA} \cite{Bed11}
Let $G=\prt{V,T,S,P}$ be a linear grammar. Then there exists a
$\lambda$-\nla $M$ such that $\maL(M)=\maL(G)$.
\end{The}

%
%
%
%
%
%
%

The algorithms in the proofs (which can be found in \cite{Bed11}) of Theorems \ref{teo-NLA-LG} and
\ref{theo-SLNF-NLA} are dual, in the sense that applying one and then the
next, we obtain the same object. Obviously, for that in
case of Theorem \ref{theo-SLNF-NLA}, the grammar must be in the
\slnf.

\subsection{$\lambda$-moves are not necessary}

A $\lambda$-\nla without  $\lambda$-transitions is called {\bf
nondeterministic linear automaton}, in short {\bf \nla}.

 The
$\lambda$-moves in the usual classes of automata are not essential
for the model. Thus, for example, $\lambda$-moves do not increase
computational power acceptance of finite automata (see, for
example, \cite{HU79}), of pushdown automata (see, for example,
\cite{Har78}) and of Turing machines (by Church-thesis). So it is
reasonable to hope that  \nla and $\lambda$-\nla have the same
acceptance power. Nevertheless, when a $\lambda$-transition in a
$\lambda$-\nla  happens between two states of different type, it is
not obvious how we can eliminate it without changing the language
accepted by the automaton.

\begin{The} \label{theo-lambda-nla-nla} \cite{Bed11}
Let $M=\prt{Q_L,Q_R,\Sigma,\delta,I,F}$ be a $\lambda$-\nla.
Then there exists a \nla $M'$ such that $\maL(M)=\maL(M')$.
\end{The}

\section{Deterministic linear automata}

There are several nonequivalent notions for ``deterministic''
linear languages (see for example \cite{HO02}). The most general
among those classes is \textbf{DL}, where a language $L\in
\mathbf{DL}$ if it can be generated by a deterministic linear
grammar.
As proved in
\cite{HO02}, \textbf{DL} is a proper subset of $\mathbf{DCFL}\cap
\mathbf{Lin}$, where \textbf{DCFL} is the class of deterministic
context-free languages and \textbf{Lin} is the class of linear
languages.

On the other hand, a \nla is \textbf{deterministic}, \dla in
short, if for each $a\in\Sigma$ and $q\in (Q_L\cup Q_R)$, $\mid
\delta(q,a)\mid\leq 1$. Thus, in a \dla $\delta(q,a)=\{q'\}$ or $\delta(q,a)=\emptyset$. 
For that reason we consider the transition  function of a \dla 
 $M=\prt{Q_L,Q_R,\Sigma,\delta,I,F}$ as a partial function from $(Q_L\cup Q_R)\times\Sigma$ into $(Q_L\cup Q_R)$.

\begin{The} \label{theo-DL-DLA}
If $L\in \mathbf{DL}$ then there is a \dla $M$ such that $\maL(M)=L$.
\end{The}

\begin{proof}
Let $G$ be a deterministic linear grammar such that $\maL(G)=L$. From Proposition \ref{pro-SLNF} there exists a 
linear grammar $\widehat{G}=\prt{V',T,S',P'}$ in \slnf such that $\maL(G)=\maL(\widehat{G})$. By the construction,  if $A\rightarrow aB$
and $A\rightarrow aC$ then, $B=C$, and if $A\rightarrow Ba$ and
$A\rightarrow Ca$ then $B=C$. Thus, $M={Q_L,Q_R,T,\delta,\{S\},F}$
where $Q_L=\{A\in V'\; :\; A$ is a right linear variable$\}$,
$Q_R=\{A\in V'\; :\; A$ is a left linear variable$\}$,  $F=\{A\;
:\; A\rightarrow \lambda\}$ and $\delta(A,a)=B$ whenever
$A\rightarrow aB$ or $A\rightarrow Ba$ in $P'$. Clearly, $M$ is a
\dla  such that $\maL(M)=L$.
\end{proof}

\begin{Cor} Let \textbf{DLin} be the class of languages accepted by some \dla. Then

$$\mathbf{DL}\subset \mathbf{DLin}$$
\end{Cor}
\begin{proof} Straightforward from Theorem \ref{theo-DL-DLA},
$\mathbf{DL}\subseteq \mathbf{DLin}$.

On the other hand, de la Higuera and Oncina in \cite{HO02} state
that the language $\{a^nb^n\; :\; n\geq 1\}\cup\{a^nc^n\; :\;
n\geq 1\}\not\in \mathbf{DL}$. Nevertheless Figure
\ref{fig-DLA-notDL} presents a \dla which accepts this language and
therefore $\mathbf{DL}\subset \mathbf{DLin}$.
\end{proof}



\begin{figure}
\centering
\begin{tikzpicture}[shorten >=1pt,node distance=2cm,auto]

\node[state,initial]  (q_0)                          {$q_0$};
\node[transition, regular polygon sides=\4, minimum size=0.9cm]          (p_1) [right of=q_0]     { $\;\;P_1\;\;$};
\node[state,accepting]  (q_1)       [above of=p_1]                   {$q_1$};
\node[state,accepting]  (q_2)       [below of=p_1]                   {$q_2$};
\node[transition, regular polygon sides=\4, minimum size=0.9cm]          (p_2) [right of=q_1]     { $\;\;\;P_2\;\;$};
\node[transition, regular polygon sides=\4, minimum size=1cm]          (p_3) [right of=q_2]     {$\;\;P_3\;\;$ };

\path[->] (q_0) edge                 node               {\small a}     (p_1)
          (p_1) edge                 node               {\small b}     (q_1)
          (p_1) edge                 node               {\small c}     (q_2)
          (q_1) edge  [bend left]   node [above]       {\small a}     (p_2)
          (p_2) edge  [bend left]   node [above]       {\small b}     (q_1)
          (q_2) edge  [bend right]   node [above]       {\small a}     (p_3)
          (p_3) edge  [bend right]   node [above]       {\small c}     (q_2);
\end{tikzpicture}
\caption{A \dla which accepts a language not in \textbf{DL}.}\label{fig-DLA-notDL}
\end{figure}
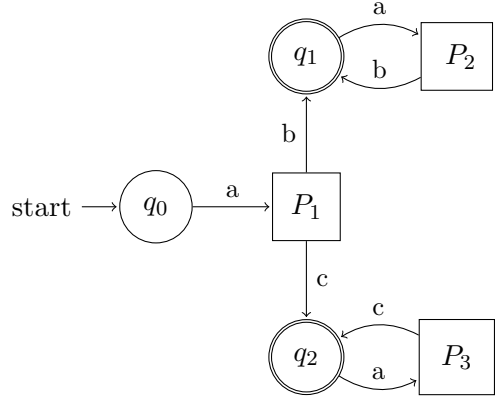


   In the following proposition we relate the class \textbf{DLin} with the class of
   deterministic context-free languages, \textbf{DCFL} in short.

\begin{Pro}\label{pro-DLL-DCFL}
$$\mathbf{DLin}- \mathbf{DCFL}\neq \emptyset\mbox{ and }\mathbf{DCFL}-\mathbf{DLin}\neq \emptyset$$

\end{Pro}

\begin{proof} The \dla in Figure \ref{fig-DLA-pal} clearly accepts the language
of palindromes over the alphabet $\{a,b\}$. But, this language is
not in \textbf{DCFL} (see \cite{HU79,LV95}). So, $\mathbf{DLin}-
\mathbf{DCFL}\neq \emptyset$. Conversely, as it is well known, the
language $L=\{a^mb^ma^nb^n\; :\; m\geq 1$ and $n\ge 1\}\in
\mathbf{DCFL}-\mathbf{Lin}$ and therefore
$\mathbf{DCFL}-\mathbf{DLin}\neq\emptyset$
\end{proof}

Since the palindromes on the alphabet $\Sigma=\{a,b\}$ is a linear language, an immediate consequence of Proposition \ref{pro-DLL-DCFL} is the following corollary.

\begin{Cor}
 $$\mathbf{DLin}- (\mathbf{DCFL}\cap \mathbf{Lin})\neq \emptyset$$
\end{Cor}

A  natural question that arises from this corollary is on the existence or not of a  language in 
$\mathbf{DCFL}\cap \mathbf{Lin}$  which is not in $\mathbf{DLin}$. We have the following conjecture:

\begin{Conj}\label{conj-DLin}
 $(\mathbf{DCFL}\cap \mathbf{Lin})-\mathbf{DLin}=\emptyset$
\end{Conj}
Case this conjecture is correct, then we will have that $(\mathbf{DCFL}\cap \mathbf{Lin})\subset \mathbf{DLin}$.



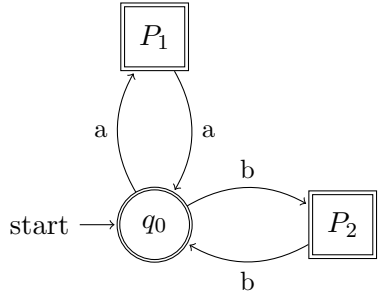
\begin{figure}
\centering
\begin{tikzpicture}[shorten >=1pt,node distance=2.5cm,auto]

\node[state,initial,accepting]  (q_0)                          {$q_0$};
\node[transition, regular polygon sides=\4, minimum size=0.9cm]          (p_1) [above of=q_0]     { $\;\;P_1\;\;$};
\node[transition, regular polygon sides=\4, minimum size=0.8cm]          (p_4) [above of=q_0]     { };
\node[transition, regular polygon sides=\4, minimum size=0.9cm]          (p_2) [right of=q_0]     { $\;\;P_2\;\;$};
\node[transition, regular polygon sides=\4, minimum size=0.8cm]          (p_3) [right of=q_0]     { };

\path[->] (q_0) edge [bend left]    node               {\small a}        (p_1)
                edge [bend left]    node               {\small b}          (p_2)   
          (p_1) edge [bend left]    node               {\small a}     (q_0)
          (p_2) edge   [bend left]  node               {\small b}     (q_0);
\end{tikzpicture}
\caption{A \dla which accepts the palindromes language.}\label{fig-DLA-pal}
\end{figure}

Each deterministic linear language is linear,  but as it will be
proved below, the converse does not hold.

%

\begin{Pro}\label{prop-DLIn-Lin} Let \textbf{Lin} be the class of linear languages.
Then

$$\mathbf{DLin}\subset \mathbf{Lin}$$
\end{Pro}

\begin{proof} Let $L(3)$ be the linear language in equation
(\ref{eq-NLL-ex}). Suppose that $L(3)\in \mathbf{DLin}$. Then,
there is a \dla $M$ such that $\maL(M)=L(3)$. Let $a^mb^n$ be a
string in $L(3)$ and $q_0\in I$. We have two cases:

In case $q_0\in Q_L$, then first $M$ must read
the leftmost $a$ in $a^mb^n$. If $\delta(q_0,a)\in Q_R$ then the
$a$ needs to match with one, two or three $b's$ what clearly
requires a nondeterministic choice and therefore is a
contradiction. If $\delta(q_0,a)=q\in Q_L$ then $M$ must read a
new $a$ and again either $\delta(q,a)\in Q_R$ in which case $M$
must read one, two or three $b's$ or $\delta(q,a)\in Q_L$ in which
case $M$ must read a new $a$, and so on. But in some moment $M$
should make a match between $a$ with one, two or three $b's$.

The case of $q_0\in Q_R$ is analogous. Therefore, $L(3)\in \mathbf{Lin}-\mathbf{DLin}$.

%
%
\end{proof}

As corollary of Prop. \ref{prop-DLIn-Lin}, we have that there are \nla for which there is no equivalent \dla, i.e. a \dla that accepts the same 
linear language. However, the next theorem provides a characterization of the \nla that are equivalent to some \dla, i.e. 
\nla which accept deterministic linear languages.

Given a \nla $M_N=\prt{Q_L,Q_R,\Sigma,\delta,I,F}$, define 
$\widetilde{Q}$ as the least set containing $\{\{q\}: q\in I\}$ and  such that 
for each $X\in \widetilde{Q}$ and $a\in\Sigma$ we have that 
 $\bigcup\limits_{q\in X} \delta(q,a)\in \widetilde{Q}$. Thus, $\widetilde{Q}\subseteq \mathcal{P}(Q_L\cup Q_R)$.

\begin{The}
 Let $L\subseteq \Sigma^*$ be a language.
 Then $L\in DLin$ iff there exists a \nla $M_N=(Q_L,Q_R,\Sigma,\delta,I,F)$ such that $\maL(M_N)=L$ and for each 
 $X\in \widetilde{Q}$, $X\subseteq Q_L$ or 
 $X\subseteq Q_R$.
\end{The}
\begin{proof}
 ($\Rightarrow$) If $L\in DLin$ then there is a \dla $M_D=\prt{Q_L,Q_R,\Sigma,\delta_D,I,F}$ such that $\maL(M_D)=L$. Let 
 $M_N=\prt{Q_L,Q_R,\Sigma,\delta,I,F}$ where $\delta(q,a)=\{\delta_D(q,a)\}$.  Clearly, $M_N$ is a \nla such that 
 $\maL(M_N)=\maL(M_D)=L$ and 
 $\widetilde{Q}=\{\{q\}: q\in Q\}$. Therefore, for 
 each $X\in\widetilde{Q}$, $X=\{q\}$ for some $q\in Q_L\cup Q_R$ and so, trivially,  $X\subseteq Q_L$ or $X\subseteq Q_R$.
 
 ($\Leftarrow$) Let $M_D=\prt{\widetilde{Q}_L,\widetilde{Q}_R,\Sigma,\delta_D,I',F'}$  where $\widetilde{Q}_L=\{X\in \wp(W): X\subseteq Q_L\}$, 
 $\widetilde{Q}_R=\{X\in \wp(W): X\subseteq Q_R\}$,  $I'=\{\{q\}: q\in I\}$, $F'=\{X\in \widetilde{Q}:X\cap F\neq \emptyset\}$ and 
 $\delta_D(X,a)=\bigcup\limits_{q\in X} \delta(q,a)$ for each $a\in \Sigma$ and $X\in \widetilde{Q}$. Clearly, $M_D$ is \dla.
 Moreover, if $w\in \maL(M_N)$, then $(q_i,w)\vdashe_{M_N} (q_f,\lambda)$ for some  $q_i\in I$ and $q_f\in F$. Therefore, if 
 $w=a_1\ldots a_n$ then there are states (possibly with repetitions) $q_{(1)},\ldots,q_{(n)}\in Q$ such that $q_{(0)}=q_i$, 
 $q_{(n)}=q_f$ and $(q_{(0)},a_1\ldots a_n)\vdash_{M_N} (q_{(1)},a_2\ldots a_n)\vdash_{M_N} \ldots \vdash_{M_N} 
 (q_{(n-1)},a_n)\vdash_{M_N} (q_{(n)},\lambda)$.
And in this case, clearly, we have that $(q_{[0]},a_1\ldots a_n)\vdash_{M_D} (q_{[1]},a_2\ldots a_n)\vdash_{M_D} \ldots \vdash_{M_D} 
 (q_{[n-1]},a_n)\vdash_{M_D} (q_{[n]},\lambda)$ when $q_{[j]}=\{q_{(j)}\}$ for each $j=0,\ldots,n$. Therefore, $w\in \maL(M_D)$. 
\end{proof}

\section{A enumerable hierarchy of linear languages}

Pushdown automata, \pda in short, have a similar characteristic to
linear automata: their nondeterministic version is more powerful
than their deterministic version. In \pda this difference allowed to
define several ways to measure nondeterminism and consequently
to determine several hierarchies of classes of context-free
languages varying from  \textbf{DCFL} into  (in
the limit) \textbf{CFL}, the class of Context-Free Languages
\cite{Bed06,Her97,SY94,VS81}\footnote{There are other hierarchies
varying from \textbf{DCFL} into \textbf{CFL}, but based in other
machine models, e.g. based on contraction automata \cite{Sol75}
and based on restarting automata \cite{Jan95,Mra01,Ott03}.}.
From that hierarchy of classes  we can
establish  a enumerable hierarchy of classes of linear languages in a simple way. Consider the
hierarchy $CFL(1), CFL(2),\ldots$ determined in \cite{Bed06}, in
this hierarchy $CFL(1)=\mathbf{DCFL}$, $CFL(k)\subset CFL(k+1)$
for each $k\geq 1$ and $\lim_{k\rightarrow \infty} CFL(k)
=\mathbf{CFL}$. So, defining $LL(k)=CFL(k)\cap \mathbf{Lin}$ we will have a enumerable
hierarchy $LL(1) \subset LL(2) \subset \ldots $ such that,
$LL(1)=\mathbf{DCFL}\cap \mathbf{Lin}$  and the limit of the
hierarchy is the class of linear languages, i.e. $\lim_{k\rightarrow \infty} LL(k) =\mathbf{Lin}$.

In the following, we will provide a enumerable hierarchy for classes of linear languages starting from 
\textbf{DLin} going, in the limits, to \textbf{Lin}.

\begin{Def}
 Let $M=\prt{Q_L,Q_R,\Sigma,\delta,I,F}$ be a \nla. The degree of explicit nondeterminism. of $M$ is the following:
 
$$Ndeg(M)=\left (\sumat{q\in Q_L\cup Q_R,a\in \Sigma} \mid\delta(q,a)\mid\right )-
\mid\{ (q,a)\in Q\times \Sigma : \delta(q,a)\neq\emptyset\}\mid$$
Let $Lin(k)$ be the class of linear languages which can be accepted by a \nla with  degree of explicit nondeterminism. $k$. 
Formally,

$$Lin(k)=\{\maL(M) : M\mbox{ is an \nla and }Ndeg(M)=k\}$$
\end{Def}

\begin{The}
 $Lin(0)=\mathbf{DLin}$, $Lin(k-1)\subset Lin(k)$ for each $k\in \mathbb{N}^+$ and 
 $\bigcup_{k\in\mathbb{N}} Lin(k)=\mathbf{Lin}$.
\end{The}
\begin{proof}
 Let $M$ be a \nla. Then $\maL(M)\in Lin(0)$ iff $Ndeg(M)=0$ iff $\mid \delta(q,a)\mid\leq 1$ iff
 $M$ is a \dla iff $\maL(M)\in \mathbf{DLin}$. Therefore $Lin(0)=\mathbf{DLin}$.
 
Let $M=\prt{Q_L,Q_R,\Sigma,\delta,I,F}$ be a \nla such that $\maL(M)\in Lin(k)$ and $Q=\{q'_1,q'_2\}$ a set of 
states such that $Q\cap (Q_L\cup Q_R)=\emptyset$. Then    $M'=\prt{Q_L\cup Q,Q_R,\Sigma,\delta',I,F}$ where for each 
$a\in \Sigma$, we have that 
$\delta'(q,a)=\delta(q,a)$ if $q\in Q_L\cup Q_R$, $\delta'(q'_0,a)=Q$ and $\delta'(q'_i,a)=\emptyset$ for each $i=1,\ldots,k$.
Clearly, $\maL(M')=\maL(M)$ and $Ndeg(M')=Ndeg(M)+1=k+1$. Thus, $\maL(M)\in Lin(k+1)$ and therefore $Lin(k)\subseteq Lin(k+1)$. 
On the other hand, for each $k\in \mathbb{N}$, the language 

$$L(k)=\{a^mb^n \; :\; m\leq n\leq (k+1)m\}$$

is accepted by the \nla of Figure \ref{fig-L(k)} and  
clearly in $L(k)\in Lin(k)$.  Moreover, it is evident that it is not possible to construct another \nla $M'$, for this same 
language, with a lesser degree of explicit nondeterminism., i.e. such that $Ndeg(M')<k$.

Let $L$ be a linear language. Then by Theorems \ref{theo-SLNF-NLA} and \ref{theo-lambda-nla-nla}, there is a \nla $M$ 
such that $\maL(M)=L$. Let $k=Ndeg(M)$ then $L\in Lin(k)$ and therefore, $L\in \bigcup_{k\in\mathbb{N}} Lin(k)$. So, 
$\bigcup_{k\in\mathbb{N}} Lin(k)=\mathbf{Lin}$.
\end{proof}

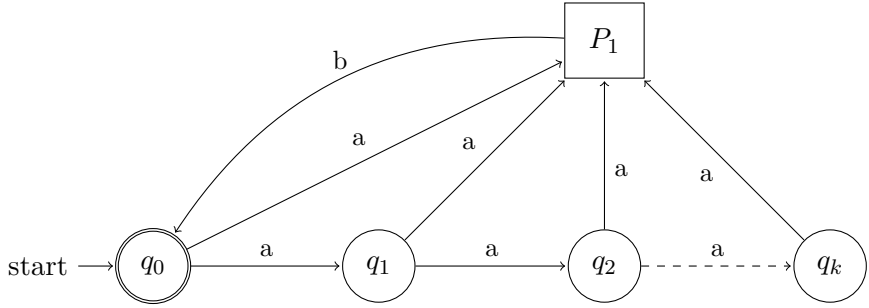
\begin{figure}
\begin{tikzpicture}[shorten >=1pt,node distance=3cm,auto]

\node[state,initial,accepting]  (q_0)                          {$q_0$};
\node[state]          (q_1) [right of=q_0]     {$q_1$};
\node[state]          (q_2) [right of=q_1]     {$q_2$};
\node[state]          (q_k) [right of=q_2]     {$q_k$};
\node[transition, regular polygon sides=\4, minimum size=1cm]          (p_1) [above of=q_2]     { $\;\;\;P_1\;\;\;$};

\path[->] (q_0) edge                 node               {\small a}        (q_1)
                edge                 node               {\small a}          (p_1)   
          (q_1) edge                 node               {\small a}     (q_2)
                edge                 node               {\small a}          (p_1)     
          (q_2) edge  [dashed]       node               {\small a}    (q_k)   
          (q_2) edge                 node [below right] {\small a}    (p_1)   
          (q_k) edge                 node               {\small a}     (p_1)
          (p_1) edge  [bend right]   node [above]       {\small b}     (q_0);
\end{tikzpicture}
\caption{\nla which accept the language $L(k)$}\label{fig-L(k)}
\end{figure}

\section{Even linear languages}

The class of even linear languages was introduced by Amar and
Putzolu in \cite{AP64}. This class properly contains   the class of
regular language and it is properly contained in the class of linear
languages. An important characteristic of this class of language 
is that it allows a solution of the learning problem\footnote{The learning problem for a class of formal languages, 
is the search of ``learning  procedures''  for  acquiring  grammars  on  the
basis of exposure to evidence about languages in the class \cite{Pul03}.} for some subclasses of even linear languages 
based on positive examples such as the class of deterministic even linear languages \cite{KMT97}. 

Basically,
a language is even linear if it is generated by an \textbf{even
linear grammar}, i.e. a linear grammar where each production of
the form $A\raw uBv$ satisfies ${\mid} u{\mid} ={\mid} v{\mid}$
\cite{KMT97,Tak88}. As it is well known, each even linear grammar has
a normal form where each production has either the form $A\raw
uBv$ or the form $A\raw a$, where $\mid u\mid =\mid v\mid=1$  and
$a\in \Sigma\cup\{\lambda\}$ \cite{KMT97}.

Let $M$ be a \nla. $M$ is an \textbf{even \nla} if its
transition diagram is a bipartite graph with  $Q_L$ and $Q_R$ as
its partitions, i.e. if for each $q\in Q_L$, $p\in Q_R$ and
$a\in \Sigma$, $\delta(q,a)\subseteq Q_R$ and
$\delta(p,a)\subseteq Q_L$. For example, the \dla in Figure
\ref{fig-DLA-pal} is an even \nla.
%
%


\begin{Pro}
Let $G$ be an even linear grammar. Then there is an even \nla $M$
such that $\maL(G)=\maL(M)$.
\end{Pro}

\begin{proof} Without loss of generality we can suppose that
$G=\prt{V,T,S,P}$ is in the even linear normal form. Let
$G'=\prt{V\cup V',T,S,P'}$ be the grammar obtained as follows:

Start with  $V'=P'=\emptyset$. For each production $A\raw aBb$ in
$P$ add a new variable $C$ to $V'$ and the productions $A\raw aC$
and $C\raw Bb$ to $P'$. Finally, add to $P'$ each production
$A\raw a$ in $P$. Clearly, $G'$ is equivalent to $G$ and it is in
\slnf.

Now, applying the algorithm in \ref{theo-SLNF-NLA}, we will have an
even \nla equivalent to $G$.
 \end{proof}

 Conversely,

\begin{Pro}
Let $M$ be an even \nla. Then there is an even linear grammar $G$
such that $\maL(G)=\maL(M)$.
\end{Pro}

\begin{proof} Applying the algorithm
 in Theorem \ref{teo-NLA-LG} to the even \nla $M$, we will obtain a linear grammar $G$ with three  kinds of productions:

 \begin{enumerate}
 \item $q\raw a p$, where $a\in \Sigma$, $q\in Q_L$ and $p\in Q_R$.

 \item $p\raw qa$, where $a\in \Sigma$, $q\in Q_L$ and $p\in Q_R$.

 \item $q\raw \lambda$, where $q\in Q_L\cup Q_R$.
 \end{enumerate}

 Now, we construct a new grammar $G'=\prt{V,T,S,P'}$ from $G=\prt{V,T,S,P}$ as follow:

 For each production $q\raw ap$ in $P$ put in $P'$ the productions $q\raw ax$ for each $p\raw x$ in $P$.
Analogously, for each production $p\raw qa$ in $P$ put in $P'$ the
productions $p\raw xa$ for each $q\raw x$ in $P$.

 Clearly, $G'$ is in the formal norm for even grammar and it is equivalent to $G$.
 \end{proof}
 

\section{Final remarks}

Since \nla is a two-read head model which works two ways, it cannot be
considered as automata model in the sense of an abstract family
of automata as done by Ginsburg \cite{Gin75}. However, \nla can be
considered as automata in the more intuitive and general notion, as for
example, ``An automaton is a device which recognizes or accepts certain elements of $\Sigma^*$,
where $\Sigma$ is a finite alphabet'' \cite{And06} or ``An
automaton is a, construct that possesses all the indispensable features of a
digital computer. It accepts inputs, produces output, may have some temporary
storage, and can make decisions in transforming the input into the
output'' \cite{Lin01}.

The normal form for linear grammars in section
\ref{sec-LG} is useful to turn easier the
proofs in section \ref{sec-NLA} and therefore does not intend to be an
alternative to the well known normal form for the linear grammars.

The contribution of this work was to provide two subclasses of  \nla, namely
\dla and even \nla, which model the subclasses of deterministic
and even linear languages, respectively. In particular, there are differents, but not equivalents,
proposed for the class of ``deterministic'' linear languages and here we proved that \textbf{DLin}, i.e. the class of 
languages accepted by some \dla, contain all those which are originated by a restriction in the linear grammars as 
considered in \cite{HO02}. In addition,  $\mathbf{DLin}$  is not a proper subset of $\mathbf{DCFL}\cap \mathbf{Lin}$, i.e. the class of linear 
languages which are also  deterministic context free languages. In fact, we conjecture that 
$\mathbf{DCFL}\cap \mathbf{Lin}\subset \mathbf{DLin}$. 
The advantage of using the automata models introduced here (\dla and even \nla) is due to their simplicity with respect to other 
automata models for these classes of languages and in the case of \dla, another advantage is that the class of 
languages modeled by this class of automata (which are naturally deterministic) is broader than the several classes of  
deterministic linear languages proposed in the literature. Other minor contribution was providing a characterization of 
\nfa which accept languages also accepted by \dla and a method to obtain this \dla and giving an enumerable hierarchy 
of linear languages starting by the class \textbf{DLin} and having as limits 
the class of linear languages.
 As a future works we can intend to prove the Conjecture \ref{conj-DLin}, compare the class \textbf{DLin} with the classes 
 of left-right determinitic linear languages  introduced in \cite{CRO04,CRO06} and study some closure property of \textbf{DLin}.





\end{document}